\documentclass[12pt]{article}
\usepackage{microtype}
\usepackage{color}
\usepackage[small]{titlesec}
\usepackage{amssymb,amsmath,amsthm,mathtools}
\usepackage{enumerate}
\usepackage{hyperref}
\usepackage{cite}
\usepackage[mathscr]{euscript}
\usepackage[letterpaper,hmargin=3.7cm,vmargin=3.5cm]{geometry}

\usepackage{setspace}
\usepackage{graphicx}
\newtheorem{theorem}{Theorem}[section]
\newtheorem{proposition}[theorem]{Proposition}
\newtheorem{lemma}[theorem]{Lemma}

\theoremstyle{definition}

\newtheorem*{acknowledgments}{Acknowledgment}
\theoremstyle{remark}
\newtheorem{remark}{Remark}

\newcommand{\abs}[1]{\left\lvert #1 \right\rvert}
\newcommand{\norm}[1]{\left\lVert #1 \right\rVert}
\newcommand{\I}{{\rm i}}
\newcommand{\inner}[2]{\left\langle#1,#2\right\rangle}

\newcommand{\cB}{\mathcal{B}}

\newcommand{\R}{{\mathbb R}}
\newcommand{\C}{{\mathbb C}}
\newcommand{\N}{{\mathbb N}}

\newcommand{\cc}[1]{\overline{#1}}

\newcommand{\mb}{\boldsymbol}

\DeclareMathOperator{\dom}{dom}
\DeclareMathOperator{\Ker}{ker}

\DeclareMathOperator{\ran}{ran}
\DeclareMathOperator{\spec}{spec}
\DeclareMathOperator{\Span}{span}

\DeclareMathOperator{\assoc}{assoc}

\begin{document}
\begin{titlepage}
\title
{\vspace{-1cm}On dB
  spaces with nondensely defined \\ multiplication
  operator and the  existence \\ of zero-free functions
\footnotetext{%
Mathematics Subject Classification(2010):
46E22; 
Secondary
47A25, 
47B25, 
47N99. 
}
\footnotetext{%
Keywords: de Branges spaces, zero-free functions.
}
\\[2mm]}
\author{
\textbf{Luis O. Silva}\thanks{Partially supported by CONACYT 
	(M{\'e}xico) through grant CB-2008-01-99100}
\\
\small Departamento de F\'{i}sica Matem\'{a}tica\\[-1.6mm]
\small Instituto de Investigaciones en Matem\'{a}ticas Aplicadas y
	en Sistemas\\[-1.6mm]
\small Universidad Nacional Aut\'{o}noma de M\'{e}xico\\[-1.6mm]
\small C.P. 04510, M\'{e}xico D.F.\\[-1.6mm]
\small \texttt{silva@iimas.unam.mx}
\\[4mm]
\textbf{Julio H. Toloza}\thanks{Partially supported by CONICET 
	(Argentina) through grant PIP 112-201101-00245}
\\
\small CONICET, and\\[-1.6mm]
\small Centro de Investigaci\'{o}n en Inform\'{a}tica para la
	Ingenier\'{i}a\\[-1.6mm]
\small Universidad Tecnol\'{o}gica Nacional --
	 Facultad Regional C\'{o}rdoba\\[-1.6mm]
\small Maestro L\'{o}pez esq.\ Cruz Roja Argentina\\[-1.6mm]
\small X5016ZAA C\'{o}rdoba, Argentina\\[-1.6mm]
\small \texttt{jtoloza@scdt.frc.utn.edu.ar}}

\date{}
\maketitle

\begin{center}
\begin{minipage}{5in}
  \centerline{{\bf Abstract}} \bigskip
  In this work we consider de Branges spaces where the multiplication
  operator by the independent variable is not densely defined.  First,
  we study the canonical selfadjoint extensions of the multiplication
  operator as a family of rank-one perturbations from the viewpoint
  of the theory of de Branges spaces.  Then, on the basis of the
  obtained results, we provide new necessary and sufficient conditions
  for a real, zero-free function to lie in a de Branges space.
\end{minipage}
\end{center}
\bigskip
\thispagestyle{empty}
\end{titlepage}

\section{Introduction}

In this paper we deal with some properties of the class of de Branges
spaces (dB spaces) characterized by the fact that the  operator $S$ of 
multiplication by the independent variable is not densely defined. 
We recall that a de Branges space $\cB$ is a Hilbert space of entire 
functions which can be defined by means of an Hermite-Biehler 
function $e(z)$ (for details see Section 2). As it is well known, 
when the domain of $S$, denoted $\dom(S)$, is not dense its 
codimension equals one \cite[Theorem~29]{debranges}.
In particular, in such case $\dom(S)$ is orthogonal to one of the 
associated functions 
\begin{equation}\label{eq:assoc-functions}
s_\beta(z):=\frac{i}{2}\left[e^{i\beta}e(z)-e^{-i\beta}e^\#(z)\right]
	     = \sin\beta\, s_{\pi/2}(z) + \cos\beta\, s_0(z),
\end{equation}
where $\beta\in[0,\pi)$. 
As we recall in Section 3, this family of functions is in one-to-one 
correspondence with the set of canonical (that is, within
$\cB$) selfadjoint extensions $S_\beta$ of $S$.  Moreover, the
function $s_\beta(z)$ that is orthogonal to $\dom(S)$ is the only one 
of this family that belongs to $\cB$. Without loss of generality we
shall henceforth assume that this occurs for $s_0(z)$ (otherwise one
can always perform a change of parametrization).

We begin by looking at the operator $S$ and its canonical selfadjoint 
extensions. The main result here is 
Theorem~\ref{thm:rank-one-perturbations}, where we render the 
selfadjoint operator extensions of $S$ as a family of rank-one 
perturbations of $S_{\pi/2}$ along the function $s_0(z)\in\cB$, viz.,
\begin{equation}
\label{eq:rank-one-perturbation-intro}
S_\beta = S_{\pi/2}-\frac{\cot\beta}{\pi}
                    \inner{s_0(\cdot)}{\cdot\,}_\cB s_0(z),
          \quad \beta\in(0,\pi).          
\end{equation}
Generically speaking, a formula of this sort is known to be valid from
the abstract theory of rank-one perturbations of relations with
deficiency indices $(1,1)$; see for instance \cite{hassi}. However, we
derive \eqref{eq:rank-one-perturbation-intro} exclusively from the
properties of functions in dB spaces and the family $s_\beta(z)$,
$\beta\in[0,\pi)$. We believe that this derivation yields further
insight on the interplay between the functions $s_\beta(z)$ and the
corresponding selfadjoint extensions of $S$. In passing, we note that
the selfadjoint extension $S_0$ is not itself an operator but rather a
(multi-valued) linear relation; see
\eqref{eq:linear-relation-selfadjoint} below.

Equation \eqref{eq:rank-one-perturbation-intro} suggests studying
whether $s_0(z)$ is a generating vector of a selfadjoint extension of
$S$. For a definition of generating vector we refer the reader to
\cite[Section 69, Definition 1]{MR1255973}.
Theorem~\ref{thm:generating-element} asserts that $s_0(z)$ is a
generating vector for $S_{\pi/2}$, and therefore, for all of the
selfadjoint extensions of $S$ with the exception of $S_0$.

With these results at hand, we tackle the question of whether a dB
space of the class considered in this work has a zero-free function.
The existence of a real zero-free function in a dB space (or more
generally, in certain spaces of functions associated to it) has been
studied in great detail; see for instance
\cite{langer-woracek,II,woracek2}. From the point of view of Krein's
theory of entire operator \cite{krein3}, if $g(z)\in\cB$ is zero-free
then it is an entire gauge for $S$, viz., it satisfies
\begin{equation}
  \label{eq:entire-condition}
  \cB=\ran(S-wI)\dotplus\Span\{g(z)\}\,,\qquad\forall w\in\C.
\end{equation}

In Theorem~\ref{thm:zero-free-function} we show that a real zero-free
function of the form $s_{\beta}(z)/j_\beta(z)$ is in $\cB$, where
$j_\beta(z)$ is any real entire function whose zero-set coincides with
that of $s_\beta(z)$, if and only if
\begin{equation*}
  \frac{1}{j_\beta(z)}=\sum_{k=1}^\infty\frac{c_k}{z-x_k},
\end{equation*}
where $\{c_k\}_{k\in\N}$ satisfies
\begin{equation*}
  \sum_{k=1}^\infty\abs{c_k}^2\abs{\frac{s'_{\beta}(x_k)}{s_0(x_k)}}
  <\infty.
\end{equation*}

Theorem~\ref{thm:zero-free-function} does not hold for $\beta=0$.
This case is treated apart in
Theorem~\ref{thm:alternative-characterization}, where specific
necessary and sufficient conditions, for a real zero-free function of
the form $s_0(z)/j_0(z)$ to be in $\cB$, are given. This
characterization is based on the fact, elaborated in Remark 4, that
every zero-free function in $\cB$ is a generating vector for some,
hence every, selfadjoint operator extension of $S$.

According to \cite[Theorem 3.2]{woracek2}, if there exists a real
zero-free function in the dB space, then this function is unique up to
a multiplicative real constant. Therefore, in this case, all the
functions $s_{\beta}(z)/j_\beta(z)$, $\beta\in[0,\pi)$, are basically
the same one. Thus, Theorems~\ref{thm:zero-free-function} and
\ref{thm:alternative-characterization} give two different
characterizations of dB spaces with nondensely defined multiplication
operator and having zero-free functions. Note also that, since
(\ref{eq:entire-condition}) means that $S$ is an entire operator, each
of these theorems provides necessary and sufficient conditions for the
operator $S$ to be entire.

It is worth remarking that the characterizations of dB spaces with
zero-free functions given by Theorems~\ref{thm:zero-free-function} and
\ref{thm:alternative-characterization} differ from all
the characterizations already known, viz., from the one stated by de
Branges \cite[Theorem~25]{debranges} and those found, with diverse 
degree of generalization, in \cite{langer-woracek,II,woracek2}.

By the end of this note we briefly address the question of how 
the generating vector $s_0(z)$ and the entire gauge $g(z)$ are related. 
Proposition~\ref{prop:gauge-perturbation} is a
simple observation on a connection between these two functions
within the dB space.

\section{Preliminaries}

In what follows by a dB space we will always mean a de Branges
Hilbert space.

The usual definition of a dB space starts from an Hermite-Biehler 
function, that is, an entire function $e(z)$ satisfying 
$\abs{e(z)}>\abs{e(\cc{z})}$ for all $z\in\C^+$. Then, the dB space 
generated by $e(z)$ is defined as 
\begin{equation*}
  \cB(e):=\{f(z)\text{ entire}: f(z)/e(z), f^\#(z)/e(z)\in H_2(\C^+)\},
\end{equation*}
where $H_2(\C^+)$ is the Hardy space
\begin{equation*}
  H_2(\C^+):=\{f(z)\text{ is holomorphic in }\C^+: \sup_{y>0}
         \int_\R\abs{f(x+\I y)}^2dx<\infty\};
\end{equation*}
here $\C^+$ denotes the open upper half-plane. We also use the standard
notation $f^\#(z):=\cc{f(\cc{z})}$. The linear space
$\cB(e)$ equipped with the inner product
\begin{equation}
  \label{eq:dB-inner-product}
  \inner{g}{f}_\cB:=\int_\R\frac{\cc{g(x)}f(x)}{\abs{e(x)}^2}dx
\end{equation}
is a Hilbert space \cite[Theorem 2.2]{remling}.

There are alternative definitions of a dB space; see for instance
\cite[Proposition 2.1]{remling} and \cite[Chapter2]{debranges}. It is
also possible to define a de Branges space without relying on a given
Hermite-Biehler function \cite[Problem~50]{debranges}. Moreover, a
given dB space can be generated by different Hermite-Biehler functions
\cite[Theorem~1]{debranges0}.

By definition, a dB space has a reproducing kernel, that is, there
exists a function $k(z,w)$ that belongs to $\cB$ for all $w\in\C$ and
satisfies the property $\inner{k(\cdot,w)}{f(\cdot)}_\cB=f(w)$ for all
$f(z)\in\cB$. Moreover, $k(w,z)=\cc{k(z,w)}$ and
$\cc{k(\cc{z},w)}=k(z,\cc{w})$ \cite[Theorem~23]{debranges}.

One important operator in a dB space is the operator of
multiplication by the independent variable,
\[
\dom(S) = \{f(z)\in\cB: zf(z)\in\cB\},\quad (Sf)(z):=zf(z).
\]
This operator is symmetric, closed, regular, and has deficiency
indices $(1,1)$. Its domain has codimension 1 or 0, depending on
whether one (and in that case, only one) of the functions $s_\beta(z)$
is within $\cB$ or none is \cite[Theorem~29]{debranges}.

To any de Branges space there corresponds a so-called space of
associated functions \cite[Section~25]{debranges}. This space can be
succinctly defined by
\begin{equation*}
\assoc\cB := \cB + z\cB
\end{equation*}
(see \cite[Lemma~4.5]{kaltenback}). Within $\assoc\cB$ lies the
distinguished family of entire functions $s_\beta(z)$ defined by
\eqref{eq:assoc-functions}. Generically,
$s_\beta(z)\in\assoc\cB\setminus\cB$. As already mentioned, this
family of functions is in bijective correspondence with the set of
canonical selfadjoint extension of $S$ (see
\eqref{eq:definition-selfadj} and
\eqref{eq:linear-relation-selfadjoint} below). From its definition, it
follows that $s_\beta(z)$ is real (that is, it satisfies
$\cc{s_\beta(z)}=s_\beta(\cc{z})$), it can also be verified that this
function has only simple zeros and its
zero-set coincides with the spectrum of the corresponding selfadjoint
extension $S_\beta$.

The reproducing kernel can be written in terms of the functions
$s_\beta(z)$. In particular \cite[Section~2]{langer-woracek},
\begin{equation}
\label{eq:reproducing-kernel}
k(z,w) = \begin{dcases}
         \frac{s_{\pi/2}(z)s_0(\cc{w})-s_0(z)s_{\pi/2}(\cc{w})}
		 {\pi(z-\cc{w})},& z\ne \cc{w},
         \\
         \frac{1}{\pi}
		 \left[s'_{\pi/2}(z)s_0(z)-s_{\pi/2}(z)s'_0(z)\right],&
		 z=\cc{w}.
           \end{dcases}
\end{equation}

\section[Selfadjoint extensions of multiplication operator]
{Selfadjoint extensions of $\mb{S}$}

Since we are assuming that the multiplication operator $S$ in $\cB$ is
not densely defined, one of the functions $s_\beta(z)$ necessarily
belongs to $\cB$. As already mentioned, we may suppose that this
happens for $\beta=0$. Consequently, $\dom(S)^\perp=\Span\{s_0(z)\}$
\cite[Theorem 29]{debranges}. The selfadjoint operator extensions of
$S$, corresponding to $\beta\in(0,\pi)$, can be described as follows
\cite[Propositions~4.5 and 6.1]{kaltenback} 
(cf. \cite[Proposition~3.8]{II}),
\begin{subequations}
\label{eq:definition-selfadj}
\begin{gather}
\dom(S_\beta) =
	\left\{g(z)=\frac{s_\beta(w)f(z)-s_\beta(z)f(w)}{\sin\beta\,(z-w)},
	\, f(z)\in\cB,\, w\in\C\right\},
	\label{eq:dom-selfadj}
\\[2mm]
(S_\beta g)(z) = z g(z) + \frac{1}{\sin\beta}f(w)s_\beta(z),
	\label{eq:action-selfaj}
\end{gather}
\end{subequations}
while the remainder selfadjoint extension of $S$ is given by the linear
relation
\begin{equation}
\label{eq:linear-relation-selfadjoint}
S_0 = \left\{(g(z),zg(z) + cs_0(z)): g(z)\in\dom(S),\, c\in\C\right\};
\end{equation}
clearly $\dom(S_0)=\dom(S)$.

\begin{lemma}
\label{lem:rank-one}
Assume $s_0(z)\in\cB$. Then $\dom(S_\beta)=\dom(S_{\pi/2})$ for all
$\beta\in(0,\pi)$. Furthermore,
\begin{equation}\label{eq:almost-rank-one}
\left(S_\beta g\right)(z)
   = \left(S_{\pi/2}g\right)(z) + \frac{\cos\beta}{\sin\beta}f(w)s_0(z),
\end{equation}
for all $g(z)\in\dom(S_\beta)$ and where $f(z)$ is related to $g(z)$ by
\eqref{eq:dom-selfadj}.
\end{lemma}
%
%
%
\begin{proof}
Consider $g(z)\in\dom(S_\beta)$. By \eqref{eq:dom-selfadj},
\[
g(z) = \frac{s_\beta(w)f(z)-s_\beta(z)f(w)}{\sin\beta\,(z-w)}
\]
for some $f(z)\in\cB$. Using \eqref{eq:assoc-functions} this can be
written as
\begin{equation}
\label{eq:decomposition-of-g}
g(z) = \frac{s_{\pi/2}(w)f(z)-s_{\pi/2}(z)f(w)}{z-w}
	+\frac{\cos\beta}{\sin\beta}\,\frac{s_0(w)f(z)-s_0(z)f(w)}{z-w}.
\end{equation}
The first term above belongs to $\dom(S_{\pi/2})$ due to
\eqref{eq:dom-selfadj}, hence the second one belongs to $\cB$.
Moreover, since $s_0(z)\in\cB$, the numerator of the second
term is in $\cB$ and, therefore, it is in $\ran(S-wI)$. Thus, the 
second term in \eqref{eq:decomposition-of-g} lies in $\dom(S)$.
Taking into account that
$\dom(S)\subset\dom(S_{\pi/2})$, one concludes that there exists
$n(z)\in\cB$ such that
\begin{equation}
\label{eq:mess}
\frac{\cos\beta}{\sin\beta}\,\frac{s_0(w)f(z)-s_0(z)f(w)}{z-w}
	= \frac{s_{\pi/2}(w)n(z)-s_{\pi/2}(z)n(w)}{z-w},
\end{equation}
The fact that both terms in (\ref{eq:decomposition-of-g}) belong to
$\dom(S_{\pi/2})$ shows that
$\dom(S_\beta)\subseteq\dom(S_{\pi/2})$. Since in the argument above
we can switch the roles of $S_\beta$ and $S_{\pi/2}$, we in fact have
$\dom(S_\beta)=\dom(S_{\pi/2})$.

Since the numerator of the l.h.s. of \eqref{eq:mess} lies in $\cB$, 
it also does the numerator of the r.h.s. As $s_{\pi/2}(z)\not\in\cB$,
necessarily $n(w)=0$, that is, $n(z)\in\ran(S-wI)$. 
Let $h(z):=f(z)+n(z)$. Then, resorting to \eqref{eq:assoc-functions} 
once again, we obtain
\begin{align*}
\left(S_\beta g\right)(z)
	&= zg(z) + \frac{1}{\sin\beta}f(w)s_\beta(z)
	\\
	&= zg(z) + f(w)s_{\pi/2}(z) + \frac{\cos\beta}{\sin\beta}f(w)s_0(z)
	\\
	&= zg(z) + h(w)s_{\pi/2}(z) - n(w)s_{\pi/2}(z)
		   + \frac{\cos\beta}{\sin\beta}f(w)s_0(z),
\end{align*}
which yields \eqref{eq:almost-rank-one}.
\end{proof}

The following assertion does not depend on assuming that a 
function $s_\beta(z)$ is in $\cB$ (that is, it holds on any dB space).

\begin{lemma}\label{lemma:things-are-orthogonal-to-s}
 For every $s_\beta(z)$, $\beta\in[0,\pi)$, and
  $h(z)\in\dom(S)$,
\[
\int_{-\infty}^\infty \frac{\cc{s_\beta(x)}h(x)}{\abs{e(x)}^2}dx = 0.
\]
\end{lemma}
\begin{proof}
Let $x_0$ be a zero of $s_\beta(z)$. Then, by 
  (\ref{eq:definition-selfadj}), $k(z):=s_\beta(z)/(z-x_0)$ is an
eigenfunction of $S_\beta$ with eigenvalue $x_0$. Therefore,
\[
\int_{-\infty}^\infty \frac{\cc{s_\beta(x)}h(x)}{\abs{e(x)}^2}dx
	= \int_{-\infty}^\infty \frac{\cc{k(x)}(x-x_0)h(x)}{\abs{e(x)}^2}dx
	= \inner{k(\cdot)}{(S-x_0I)h(\cdot)}_{\cB}
	= 0
\]
where the last identity follows after realizing that
$k(z)\in\Ker(S^*-x_0I)$ and the multivalued part of $S^*$ equals
$\Span\{s_0(z)\}$ \cite{hassi}.
\end{proof}

\begin{lemma}
\label{lem:perturbation-rank-one}
Let $s_0(z)\in\cB$. For $g(z)\in\dom(S_\beta)$, $f(z)\in\cB$ and
$w\in\C$ related to each other by \eqref{eq:dom-selfadj}, one has
$\inner{s_0(\cdot)}{g(\cdot)}_\cB=-\pi f(w)$.
\end{lemma}
\begin{proof}
Let us start by considering \eqref{eq:decomposition-of-g}. As already
mentioned, the second term of this identity lies in $\dom(S)$, so
\begin{align*}
\inner{s_0(\cdot)}{g(\cdot)}_\cB
	=& \int_{-\infty}^\infty
		\frac{\cc{s_0(x)}s_{\pi/2}(w)f(x)-\cc{s_0(x)}s_{\pi/2}(x)f(w)}
			{\abs{e(x)}^2(x-w)}\,dx
	\\[1mm]
	=& - \pi\inner{k(\cdot,w)}{f(\cdot)}_\cB
	\\[1mm]
	 &  + \int_{-\infty}^\infty
			\frac{\cc{s_{\pi/2}(x)}}{\abs{e(x)}^2}
		\left[\frac{s_0(w)f(x)-s_0(x)f(w)}{x-w}\right]dx,
\end{align*}
where the fact that the functions $s_\beta(z)$ are real has been used. 
In the second term, the expression between squared brackets lies in 
$\dom(S)$ so by Lemma~\ref{lemma:things-are-orthogonal-to-s} this term 
equals zero. Obviously the first term equals $-\pi f(w)$ so the 
assertion is proven.
\end{proof}

\begin{theorem}
\label{thm:rank-one-perturbations}
Assume $s_0(z)\in\cB$. Then the set of canonical selfadjoint operator
extensions of $S$ are given by
$\dom(S_\beta)=\dom(S_{\pi/2})$,
\begin{equation}
\label{eq:rank-one-perturbation}
S_\beta = S_{\pi/2}-\frac{\cot\beta}{\pi}
                    \inner{s_0(\cdot)}{\cdot\,}_\cB s_0(z),
\end{equation}
for $\beta\in(0,\pi)$.
\end{theorem}
\begin{proof}
The assertion follows straightforwardly from Lemmas \ref{lem:rank-one} 
and \ref{lem:perturbation-rank-one}.
\end{proof}

The previous discussion generalizes effortless if one assumes that
$s_\gamma(z)\in\cB$, $\gamma\in[0,\pi)$. For
$\beta\in[\gamma,\gamma+\pi)$, it is easy to see that
\[
s_\beta(z) = \sin(\beta-\gamma)s_{\gamma+\pi/2}(z)
			+ \cos(\beta-\gamma) s_{\gamma}(z).
\]
so instead of \eqref{eq:rank-one-perturbation} one has
\[
S_\beta = S_{\gamma+\pi/2}
		- \frac{\cot(\beta-\gamma)}{\pi}ç
			\inner{s_\gamma(\cdot)}{\cdot\,}_\cB s_\gamma(z),
\]
now for $\beta\in(\gamma,\gamma+\pi)$.

Now we turn to the proof of the fact that $s_0(z)$ is generating
element (see \cite[Section~69]{MR1255973}) of $S_{\pi/2}$.

\begin{theorem}
  \label{thm:generating-element}
  Assume that $s_0(z)\in\cB$. Then, for every
  $\beta\in(0,\pi)$,  $s_0(z)$ is a generating vector
  for the operator $S_{\beta}$.
\end{theorem}
\begin{proof}
  Since $s_0(z)\in\cB$,
one has, on the basis of (\ref{eq:dom-selfadj})
  and (\ref{eq:action-selfaj}), that
\begin{align*}
(S_{\pi/2}-wI)^{-1}s_0(z)
	&= \frac{1}{s_{\pi/2}(w)}
        \frac{s_{\pi/2}(w)s_0(z)-s_{\pi/2}(z)s_0(w)}{z-w}
    \\[1mm]
    &= - \frac{\pi}{s_{\pi/2}(w)} k(z,\cc{w})
\end{align*}
for all $w\not\in\spec(S_{\pi/2})$. Hence, taking into account that
(see \cite[Section~3]{II})
\begin{equation}
  \label{eq:k-in-kernel}
  k(z,\cc{w})\in\ker(S^*-wI)\quad\text{for any } w\in\C\,,
\end{equation}
one
verifies
\begin{equation*}
  \cc{\Span_{w\in\spec(S_0)}\{(S_{\pi/2}-wI)^{-1}s_0(z)\}}=\cB\,.
\end{equation*}
Thus, $s_0(z)$ is a generating element for $S_{\pi/2}$, but then, it
can be derived from (\ref{eq:rank-one-perturbation}), that $s_0(z)$ is
a generating vector for $S_\beta$ with $\beta\in(0,\pi)$.
\end{proof}
\begin{remark}
  \label{rem:s-generating}
  Alternatively, $s_0(z)$ is a generating vector
  for $S_\beta$, $\beta\in(0,\pi)$, because it has a nonzero
  projection to each eigenspace of $S_\beta$. Indeed, this follows
  from (\ref{eq:k-in-kernel}) and the fact that the eigenvalues of
  $S_\beta$ with $\beta\in(0,\pi)$ never intersect the zeros of
  $s_0(z)$.  In passing, it is also clear that $s_0(z)$ is not a
  generating vector for $S_0$.
\end{remark}

\section{On the existence of a zero-free function}

Let $E_\beta(t)$ be the resolution of the identity of $S_\beta$,
$\beta\in(0,\pi)$. Define the family of spectral functions
\begin{equation*}
      m_\beta(t):=\inner{s_0(\cdot)}{E_\beta(t)s_0(\cdot)}_\cB=
\sum_{x_n<t}\frac{\abs{s_0(x_n)}^2}{\norm{k(\cdot,x_n)}_\cB^2}\,,\quad
\{x_n\}_{n\in\N}=\spec(S_\beta)\,.
\end{equation*}
Since $s_0(z)$ is a generating element of $S_\beta$,
$\beta\in(0,\pi)$, one can consider the family of canonical isometries
$U_\beta:L_2(\R,m_\beta)\stackrel{\text{onto}}{\to}\cB$ 
(cf. \cite[Section~69, Theorem~2]{MR1255973}) given by
\begin{equation}
  \label{eq:canonical-isometry}
  (U_\beta f)(z):=f(S)s_0(z)=  \sum_{x_n\in\spec(S_\beta)}
\frac{f(x_n)\inner{k(\cdot,x_n)}{s_0(\cdot)}_\cB}
{\norm{k(\cdot,x_n)}_\cB^2}k(z,x_n)\,.
\end{equation}

\begin{theorem}
\label{thm:zero-free-function}
Assume $s_0(z)\in\cB$ and fix $\beta\in(0,\pi)$. Let $j_\beta(z)$ be
any real entire function with simple zeros exactly at
$\{x_n\}_{n\in\N}=\spec(S_{\beta})$.  The zero-free function
$s_{\beta}(z)/j_\beta(z)$ is in $\cB$ if and only if the reciprocal of
the function $j_\beta(z)$ can be decomposed as follows,
\begin{equation}
  \label{eq:decomposition}
  \frac{1}{j_\beta(z)}=\sum_{k=1}^\infty\frac{c_k}{z-x_k},
\end{equation}
where $\{c_k\}_{k\in\N}$ satisfies
\begin{equation}
\label{eq:c-numbers}
  \sum_{k=1}^\infty\abs{c_k}^2\abs{\frac{s'_{\beta}(x_k)}{s_0(x_k)}}
  <\infty
\end{equation}
and the convergence in (\ref{eq:decomposition}) is uniform on compact
subsets of $\C\setminus\spec(S_{\beta})$.
\end{theorem}
\begin{proof}
  We begin by proving the necessity of the condition for
  $\beta=\pi/2$. Since $s_0(z)$ is a generating vector for the
  operator $S_{\pi/2}$, for every $f\in L_2(\R,m_{\pi/2})$,
  $f(S)s_0(z)$ is an element of $\cB$, and any vector in $\cB$ can be
  written in this way.  Using the properties of the reproducing kernel
  and
\begin{equation}
  \label{eq:norm-of-kernel}
  \norm{k(\cdot,x_n)}_\cB^2
  	=  \inner{k(\cdot,x_n)}{k(\cdot,x_n)}_\cB
  	= -\frac{1}{\pi}s_{\pi/2}'(x_n)s_0(x_n)\,,
\end{equation}
which is obtained from~(\ref{eq:reproducing-kernel}),
one can rewrite the action of $U_{\pi/2}$ as follows
\begin{equation*}
  (U_{\pi/2} f)(z)=-\pi\sum_{x_n\in\spec(S_{\pi/2})}
\frac{f(x_n)k(z,x_n)}{s_{\pi/2}'(x_n)}\,.
\end{equation*}

Suppose that $s_{\pi/2}(z)/j_{\pi/2}(z)$ is in $\cB$, then there is a
function $f\in L_2(\R,m_{\pi/2})$ such that
\begin{align}
  \frac{s_{\pi/2}(z)}{j_{\pi/2}(z)}&=
   -\pi\sum_{x_n\in\spec(S_{\pi/2})}
\frac{f(x_n)}{s_{\pi/2}'(x_n)}k(z,x_n)\nonumber\\
&= -\sum_{x_n\in\spec(S_{\pi/2})}
\frac{f(x_n)s_{\pi/2}(z)s_0(x_n)}{s_{\pi/2}'(x_n)(z-x_n)}\,,
\label{eq:convergence-in-B}
\end{align}
where we have used (\ref{eq:reproducing-kernel}). Hence, one has
\begin{equation*}
  \frac{1}{j_{\pi/2}(z)}=-\sum_{x_n\in\spec(S_{\pi/2})}
\frac{f(x_n)s_0(x_n)}{s_{\pi/2}'(x_n)(z-x_n)}\,,
\end{equation*}
where the series converges uniformly on compacts of
$\C\setminus\spec(S_{\pi/2})$ since (\ref{eq:convergence-in-B})
converges in $\cB$. By setting
\begin{equation*}
  c_n=-\frac{f(x_n)s_0(x_n)}{s_{\pi/2}'(x_n)}
\end{equation*}
one establishes the necessity of the condition.

Let us now prove that the condition is sufficient for $\beta=\pi/2$. 
For any $n\in\N$, define
\begin{equation*}
  a_n:=\frac{c_ns'_{\pi/2}(x_n)}{s_0(x_n)}
\end{equation*}
and substitute it into (\ref{eq:decomposition}) to obtain
\begin{equation*}
    \frac{1}{j_{\pi/2}(z)}=\sum_{n=1}^\infty
\frac{a_ns_0(x_n)}{s_{\pi/2}'(x_n)(z-x_n)}\,.
\end{equation*}
Therefore, using (\ref{eq:reproducing-kernel}) and
(\ref{eq:norm-of-kernel}), one has
\begin{align}
\frac{s_{\pi/2}(z)}{j_{\pi/2}(z)}
&= \sum_{n=1}^\infty
   \frac{a_ns_{\pi/2}(z)s_0(x_n)}{s_{\pi/2}'(x_n)(z-x_n)}\nonumber
\\
&= \sum_{n=1}^\infty
   \frac{a_n\inner{k(\cdot,x_n)}{s_0(\cdot)}}
        {\norm{k(\cdot,x_n)}^2}k(z,x_n)
   \label{eq:actual-canonical-repr}
\end{align}
for any $z\in\C$.
By definition of the numbers $\{a_n\}_{n\in\N}$ there is a function
$f\in L_2(\R,m_{\pi/2})$ such that $f(x_n)=a_n$ for all $n\in\N$. 
Thus, (\ref{eq:actual-canonical-repr}) means that
$s_{\pi/2}(z)/j_{\pi/2}(z)=(U_{\pi/2}f)(z)\in\cB$.

Once the assertion has been proven for $\beta=\pi/2$, one uses
\cite[Lemmas~3.3 and 3.4]{II} to finish the proof.
\end{proof}
\begin{remark}
A. Baranov pointed out to us that Theorem~\ref{thm:zero-free-function} 
for $\beta=\pi/2$ can be proven by expanding the function 
$s_{\pi/2}(z)/j_{\pi/2}(z)$ with respect to the orthonormal basis 
$k(z,x_n)/\norm{k(\cdot,x_n)}$ 
(with $\{x_n\}_{n\in\N}=\spec(S_{\pi/2})$), thus obviating the use 
of a generating vector.
\end{remark}
\begin{remark}
  If (\ref{eq:decomposition}) and (\ref{eq:c-numbers}) hold, and
  additionally we suppose that
  \begin{equation*}
    \abs{c_n}(1+\abs{x_n})\ge\abs{\frac{s_0(x_n)}{s_{\pi/2}(x_n)}}
    \quad\text{for
    all }n\in\N\,,
  \end{equation*}
then, due to a theorem by Krein \cite[Lecture~16, Theorem~3]{levin}, 
the function $s_\beta(z)/j_\beta(z)$ is in the Cartwright class.
\end{remark}
\begin{remark}
  Clearly, if a zero-free function belongs to $\cB$, then it is a
  generating vector for $S_\beta$ with $\beta\in[0,\pi)$, since it has
  a nonzero projection onto every eigenspace
  (cf. Remark~\ref{rem:s-generating}).
\end{remark}

By using the fact that $s_0(z)/j_0(z)$ is a generating vector for any
selfadjoint extension whenever $s_0(z)/j_0(z)\in\cB$, we prove the
following assertion which gives a different set of necessary and
sufficient conditions for a zero-free function to be in $\cB$.

\begin{theorem}
\label{thm:alternative-characterization}
Assume $s_0(z)\in\cB$ and let $j_\beta(z)$ be defined as in
Theorem~\ref{thm:zero-free-function}. If the function $s_0(z)/j_0(z)$ 
is in $\cB$, then, for all $\beta\in(0,\pi)$,
\begin{equation}
\label{eq:necessity-alt}
\frac{1}{j_0(t)}\in L_2(\R,m_\beta)\quad\text{and}\quad
\frac{s_0(z)}{j_0(z)}=(U_\beta\,\frac{1}{j_0})(z)\,.
\end{equation}
Conversely, if there exists a set $\mathcal{I}\subset(0,\pi)$ having
an accumulation point and such that
\begin{equation*}
\frac{1}{j_0(t)}\in
L_2(\R,m_{\beta})\quad\forall\beta\in\mathcal{I}\quad
\text{and}\quad
(U_\beta\,\frac{1}{j_0})(z)=(U_{\beta'}
\frac{1}{j_0})(z)\quad\forall\beta,\beta'\in\mathcal{I}\,,
\end{equation*}
then $s_0(z)/j_0(z)$ is in $\cB$.
\end{theorem}
\begin{proof}
Assume that $s_0(z)/j_0(z)\in\cB$. Define the spectral functions
\[
\widetilde{m}_\beta(t)
  :=\inner{s_0(\cdot)/j_0(\cdot)}{E_\beta(t)s_0(\cdot)/j_0(\cdot)}_\cB,
\]	
and the isometries $\widetilde{U}_\beta$ from
$L_2(\R,\widetilde{m}_\beta)$ onto $\cB$ such that
$(\widetilde{U}_\beta f)(z):=f(S_\beta)\frac{s_0(z)}{j_0(z)}$. Since 
the function $g(t)\equiv 1$ lies in $L_2(\R,\widetilde{m}_\beta)$ for
all $\beta\in(0,\pi)$ we have the first part of 
\eqref{eq:necessity-alt}. Moreover, taking into account 
\eqref{eq:canonical-isometry}), one has
\begin{equation*}
\frac{s_0(z)}{j_0(z)}=(\widetilde{U}_\beta g)(z)
	= \sum_{x_n\in\spec(S_\beta)}\frac{s_0(x_n)}
		 {j_0(x_n)\norm{k(\cdot,x_n)}_\cB^2}k(z,x_n)
	=(U_\beta\frac{1}{j_0})(z)
\end{equation*}
for every $\beta\in(0,\pi)$.  For the converse part of the assertion,
consider the function $r(z)=(U_\beta\frac{1}{j_0})(z)$ for all 
$\beta\in\mathcal{I}$. It is straightforward to verify that
$r(x_n)=s_0(x_n)/j_0(x_n)$ for $x_n\in\spec(S_\beta)$. Now, since
$\mathcal{I}$ has an accumulation point,
\cite[Chapter 7, Theorem 3.9]{MR0407617} implies that the entire
functions $r(z)$ and $s_0(z)/j_0(z)$ coincide in a set having
accumulation points.
\end{proof}

In \cite[Proposition~3.9]{II} (see also \cite[Theorem~3.2]{woracek2}),
necessary and sufficient conditions for a function to be in $\cB$ are
given in terms of the spectra of two selfadjoint extensions of $S$.
Two of these conditions imply that the products below are convergent
\begin{equation*}
  h_\beta(z):=
\begin{cases}
	\displaystyle{\lim_{r\to\infty}\prod_{|b_k|\le r}
			\left(1-\frac{z}{b_k}\right)}
				& \mbox{ if }0\not\in\spec(S_\beta),
			\\
			\displaystyle{z\lim_{r\to\infty}\prod_{0<|b_k|\le r}
			\left(1-\frac{z}{b_k}\right)}
				& \mbox{ otherwise, }
\end{cases}
\end{equation*}
for any $\beta\in [0,\pi)$. Moreover, the unique real zero-free
function in $\cB$ (up to a multiplicative real constant) is
$s_\beta(z)/h_\beta(z)$. Therefore, one arrives at the following
straightforward conclusion.
\begin{proposition}
  \label{prop:rel-with-woracek}
  Let $s_0(z)$ be an element of $\cB$. If
  $s_\beta(z)/j_\beta(z)\in\cB$, then $j_\beta(z)=h_\beta(z)$. On the
  other hand, if $j_\beta(z)$ is decomposed as in
  (\ref{eq:decomposition}) with the sequence $\{c_n\}_{n\in\N}$
  satisfying (\ref{eq:c-numbers}), then $j_\beta(z)=h_\beta(z)$.
\end{proposition}
\begin{remark}
  Assuming that $\cB$ is decomposed as in (\ref{eq:entire-condition})
  (equivalently that there is a zero-free function in $\cB$), the
  unique real zero-free function is nothing but the unique real entire
  gauge (up to a multiplicative real constant).
\end{remark}
In order to clarify the connection between the gauge and the function
$s_0(z)$, let us define the operator $f(S)$ as the operator in $\cB$
given by
\begin{equation*}
  \dom(f(S)):=\{g(z)\in\cB : f(z)g(z)\in\cB\}\,,\quad
  (f(S)g)(z):=f(z)g(z)\,.
\end{equation*}
Clearly this definition is consistent with the notion of a function of
an operator. Moreover, the following assertion immediately follows from
it.
\begin{proposition}
\label{prop:gauge-perturbation}
If there is a zero-free function in $\cB$, then
$s_0(z)/h_0(z)\in\dom(h_0(S))$, $s_0(z)\in\dom((1/h_0)(S))$, and
\begin{equation*}
\left(h_0(S)\frac{s_0}{h_0}\right)(z)=s_0(z)\,,\quad
\left(\frac{1}{h_0}(S)s_0\right)(z)=\frac{s_0(z)}{h_0(z)}\,.
\end{equation*}
\end{proposition}

\begin{acknowledgments}
Part of this work was done while the second author (JHT)
visited IIMAS--UNAM in the winter of 2013. He sincerely thanks them
for their kind hospitality. The authors thank A. Baranov for his
comment on Theorem~\ref{thm:zero-free-function}.
\end{acknowledgments}

\end{document}